\documentclass{elsarticle}
\journal{Operations Research Letters}

\biboptions{sort}

\usepackage{geometry}
\usepackage{array}
\usepackage{setspace}

\usepackage{enumitem}		
\usepackage[shortcuts]{extdash}

\usepackage{url}         

\usepackage{amsmath,amsfonts,amssymb,amstext,amsthm,bbm} 

\usepackage{version}
\includeversion{fullversion}\excludeversion{shortversion}

\newcommand{\nullitem}{\ensuremath{{\boldsymbol \emptyset}}}

\def\mconc{$M^{\natural}$-concave}
\def\mconcy{$M^{\natural}$-concavity}

\theoremstyle{definition}
\newtheorem{definition}{Definition}[section]
\newtheorem{remark}{Remark}[section]
\newtheorem{example}{Example}[section]

\theoremstyle{plain}
\newtheorem{theorem}{Theorem}
\newtheorem{lemma}{Lemma}[section]

\begin{document}

\begin{frontmatter}{
}

\title{Demand-Flow of Agents with Gross-Substitute Valuations
}

\author{Erel Segal-Halevi\footnote{erelsgl@gmail.com . Corresponding author.},
Avinatan Hassidim\footnote{avinatanh@gmail.com}
and Yonatan Aumann\footnote{aumann@cs.biu.ac.il}
}

\address{Bar-Ilan University, Ramat-Gan 5290002, Israel}
\begin{abstract}
We consider the gross-substitute (GS) condition introduced by Kelso and Crawford (1982). GS is a condition on the demand-flow in a specific scenario: some items become more expensive while other items retain their price. We prove that GS is equivalent to a much stronger condition, describing the demand-flow in the general scenario in which all prices may change: the demand of GS agents always flows (weakly) downwards, i.e, from items with higher price-increase to items with lower price-increase.
\\
\\
\copyright{} 2016. This manuscript version is made available under the CC-BY-NC-ND 4.0 license.

\end{abstract}

\end{frontmatter}{}
	
\section{Introduction}
Many markets involve a set of distinct indivisible goods that can be bought and sold for money. The analysis of such markets crucially depends on the agents' \emph{valuation functions} --- the functions that assign monetary values to bundles. It is common to assume that agents' valuations are weakly increasing (more goods mean weakly more value) and quasi-linear in money. Even so, without further restrictions on the valuations, the market may fail to have desirable properties such as the existence of a price-equilibrium. \citet{Kelso1982Job} introduced a property of valuations which they called \emph{gross-substitutes (GS)}. An agent's valuation has the GS property if, when the prices of some items increase, the agent does not decrease its demand for the other items. \citet{Kelso1982Job} proved that a market in which all agents are GS always has a price-equilibrium. \citet{Gul1999Walrasian,Gul2000English} complemented this result by proving that the GS condition is, in some sense, necessary to ensure existence of a price-equilibrium. The GS condition has been widely used in the study of matching markets \cite{Roth1992TwoSided}, auctions \cite{Milgrom2000Putting} and algorithmic mechanism design \cite{2007Algorithmic}.

The GS condition specifies the behavior of an agent in a very specific situation: some items become more expensive, while other items retain their original price. 
In this paper we characterize the behavior of GS agents in the more general situation, in which the prices of all items change in different ways and in different directions. This characterization may have several potential applications:

(a) Analyzing the response of markets to exogenous shocks. For example, 
suppose the government puts price-ceilings on several items. 
With a single item-type, it is obvious that a price-ceiling below the equilibrium-price will result in excess demand. But with multiple item-types, this is not necessarily so.  For example, it is possible that the prices of both item x and item y are below their equilibrium prices, but because of substitution effects, buyers switch from demanding y to demanding x so the net effect is an excess supply in y and an excess demand in x. In order to analyze such markets, we have to understand how exactly agents move from one item-type to another when the prices change. 

(b) designing dynamic combinatorial auctions. In such an auction, the auctioneer modifies the prices of different items in different rates in an attempt to change the aggregate demand. \citet{Gul2000English} describe one such auction, in which the prices are always ascending. In order to design different auctions, it may be useful to know the effect of different price-changes on the agents' demand. 

(c) Using field-data to detect the existence of complementarities (i.e, valuations that are \emph{not} GS) by comparing demands under different price-vectors.

(d) Our original application \citep{SegalHalevi2016MIDA} was a double-auction mechanism where the market-prices are set by the auctioneer in a way that guarantees truthfulness but might not be entirely efficient; understanding the demand-flow of agents let us calculate an upper bound on the loss of efficiency.

Consider two price-vectors: old and new. For every item $x$, define $\Delta_x$ as the price-increase of $x$ (the new price minus the old price). Add a "null item" \nullitem{} and set its price-increase to 0. Arrange the items vertically by ascending price-increase. Then, our main result is that:
\begin{center}
	\emph{\bf The demand of a gross-substitute agent always flows weakly downwards}.
\end{center}
I.e, an agent may switch from wanting an item whose price increased more to an item whose price increased less, but not vice-versa. This property is trivially true for a unit-demand agent, but it is not true when the agent regards some items as complementaries.
\begin{example}\label{exm:negative}
There are three items: x,y,z. Initially their prices are \$10,\$10,\$10. Then, the prices increase by $\Delta_x=\$20,\Delta_y=\$30,\Delta_z=\$40$, so that the new prices are \$30,\$40,\$50. Consider two agents with the following valuations:

\begin{center}
\begin{tabular}{|c|c|c|c|c|c|c|c|}
	\hline       & x & y & z & x+y & x+z & y+z & x+y+z \\ 
	\hline Alice & \$65 & \$70 & \$75 & \$70 & \$75 & \$75 & \$75 \\ 
	\hline Bob   & \$40 & \$40 & \$66 & \$80 & \$75 & \$75 & \$80 \\ 
	\hline 
\end{tabular} 
\end{center}
Alice has unit-demand: she needs only one item and values each bundle as the maximum item in that bundle. Bob regards x and y as complementaries: 
each of them alone is worth less than z, but together they are worth more than x+z and y+z (Note that Bob's valuation is submodular but not GS).

In the initial prices Alice's preferred bundle is z, and after the price-change her preferred bundle is x, so her demand flows downwards --- towards the smaller price-increase.

In contrast, Bob's demand is initially x+y, and after the price-change his demand is z, so his demand flows upwards --- towards the item with the larger price-increase. \qed
\end{example}
Our main result is that GS agents behave like unit-demand agents in this regard: their demand flows only downwards.

\section{Model and Notation}
There is a finite set of indivisible items, $M=\{1,\dots,m\}$. There is an $m$-sized price-vector $p$: a price per item. The price of a bundle is the sum of the prices of the items in it: $p(X):=\sum_{x\in X}p_x$.

The present paper focuses on a single agent with a single valuation-function $u: 2^M \to \mathbb{R}$. $u$ is assumed to be weakly-increasing: if a bundle $X\subseteq Y$ then $u(X)\leq u(Y)$. 

The agent's utility is quasi-linear in money. Given a utility function $u$ and a price-vector $p$, the agent's net-utility function $u_p$ is: $u_p(X) := u(X)-p(X)$.

\begin{definition}\label{def:p-demand}
Given a valuation function $u$ and a price-vector $p$, we say that a bundle $P$ is a \textbf{$p$-demand} if it is optimal for the agent to buy this bundle when the prices are $p$, i.e, the set $P$ maximizes the net-utility function $u_p(\cdot)$ over all bundles of items: $\forall X: u_p(P)\geq u_p(X)$.
\end{definition}
\begin{definition}\label{def:p-demanded}
	Given a valuation function $u$ and a price-vector $p$, we say that an item $x$ is \textbf{$p$-demanded} if there exists a $p$-demand $P$ such that $P\ni x$.
\end{definition}
\begin{definition}\label{def:abandon-discover}
	Given an agent, an item x and a pair of price-vectors $(p,q)$, we say that:
	
	(a) The agent \textbf{abandoned} item x if x is $p$-demanded but not $q$-demanded.
	
	(b) The agent \textbf{discovered} item x if x is $q$-demanded but not $p$-demanded.
\end{definition}

\begin{definition}[\cite{Kelso1982Job}]\label{def:gs}
An agent's valuation function has the \textbf{gross-substitute (GS)} property if, for every price-vectors $(p,q)$ such that $\forall y: \Delta_y\geq 0$, if $\Delta_x=0$ then 
the agent has not abandoned x.
\end{definition}

\begin{definition}\label{def:io}
A valuation has the \textbf{downward-demand-flow (DDF)} property if the following are true for every pair of price-vectors $(p,q)$ (where $\Delta_x := q_x-p_x$):

(a) If $\Delta_x\leq 0$ and the agent abandoned x, then he discovered some y with $\Delta_y<\Delta_x$.
		
(b) If $\Delta_x\geq 0$ and the agent discovered x, then he abandoned some y with $\Delta_y>\Delta_x$.
\end{definition}
\noindent
DDF implies GS: part (a) of the DDF definition implies the GS definition.
Our main result is the converse implication: GS implies DDF.

\section{$M^\natural$-concavity}
Our main technical tool is the following characterization of GS valuations \citep{Fujishige2003Note}:
\begin{definition}\label{def:mconc}
A valuation function $u$ is \emph{\mconc{}} if-and-only-if, for every two bundles $X,Y$ and for every $X' \subseteq X\setminus Y$ with $|X'|=1$ (i.e, X' is a singleton), there exists a subset $Y'\subseteq Y\setminus X$ with $|Y'|\leq 1$ (i.e, Y' is either empty or a singleton) such that: 
\begin{align*}
u(X\setminus X' \cup Y') + u(Y\setminus Y' \cup X') \geq u(X)+u(Y)
\end{align*}
\end{definition}

\begin{lemma}[\citep{Fujishige2003Note}]
A valuation function $u$ is \mconc{} if-and-only-if it is gross-substitute.
\end{lemma}

\begin{shortversion}
Using the \mconcy{} characterization, it is easy to prove that GS is preserved in net-utility functions and marginal-valuation functions:
\begin{lemma}\label{lem:mconc-net}
Let $p$ be an arbitrary price vector. A valuation function $u$ is \mconc{} if-and-only-if the net-utility function $u_p$ is \mconc{}.
\end{lemma}
\begin{definition}
Given a valuation $u$ and a constant bundle $Z$, the \textbf{marginal valuation} $u_{Z+}$ is a function that returns, for every bundle $X$ that does not intersect $Z$, the additional value that an agent holding $Z$ gains from having $X$:
\begin{align*}
u_{Z+}(X) := u(Z\cup X)-u(Z) && \text{for all $X$ with $X\cap Z=\emptyset$}
\end{align*}
\end{definition}
\begin{lemma}[\cite{Ostrovsky2015Gross}]\label{lem:mconc-marginal}
A valuation function $u$ is GS if-and-only-if, for every bundle $Z$, the marginal-valuation function $u_{Z+}$ is GS.
\end{lemma}
\end{shortversion}
\begin{fullversion}
\noindent
Below we prove that \mconcy{} is preserved in net-utility and marginal-valuation functions.

\begin{lemma}\label{lem:mconc-net}
	Let $p$ be an arbitrary price vector. A valuation function $u$ is \mconc{} if-and-only-if the net-utility function $u_p$ is \mconc{}.
\end{lemma}
\begin{proof}
	The price function $p(\cdot)$ is additive. Hence, for all $X'\subseteq X\setminus Y$ and $Y'\subseteq Y\setminus X$:
	\begin{align*}
	p(X\setminus X'\cup Y') + p(Y\setminus Y'\cup X') = p(X) + p(Y)
	\end{align*}
	Hence, the \mconc{} condition is independent of price:
	\begin{align*}
	u(X\setminus X'\cup Y') + u(Y\setminus Y'\cup X') \geq& u(X)+u(Y)
	\\
	\iff&
	\\
	u_p(X\setminus X'\cup Y') + u_p(Y\setminus Y'\cup X') \geq& u_p(X)+u_p(Y)
	\end{align*}
\end{proof}

\begin{definition}
	Given a valuation $u$ and a constant bundle $Z$, the \textbf{marginal valuation} $u_{Z+}$ is a function that returns, for every bundle $X$ that does not intersect $Z$, the additional value that an agent holding $Z$ gains from having $X$:
	\begin{align*}
	u_{Z+}(X) := u(Z\cup X)-u(Z) && \text{for all $X$ with $X\cap Z=\emptyset$}
	\end{align*}
\end{definition}

\begin{lemma}\label{lem:mconc-marginal}
	A valuation function $u$ is \mconc{} if-and-only-if, for every bundle $Z$, the marginal-valuation function $u_{Z+}$ is \mconc{}.
\end{lemma}
\begin{proof}
	The "if" direction is obvious since $\emptyset$ is also a bundle and $u_{\emptyset+}\equiv u$.
	
	For the "only if" direction, suppose $u$ is \mconc{} and let $Z$ be an arbitrary bundle. We have to prove that $u_{Z+}$ is \mconc{}, i.e, for all bundles $X,Y$ with $X\cap Z=Y\cap Z=\emptyset$, and for every $X'\subseteq X\setminus Y$ with $|X'|=1$, there exists a $Y'\subseteq Y\setminus X$ with $|Y'|\leq 1$ such that:
	\begin{align}\label{eq:uc-is-mconc}
	u(Z\cup (X\setminus X'\cup Y')) + u(Z\cup (Y\setminus Y'\cup X')) \geq& u(Z\cup X)+u(Z \cup Y)
	\end{align} 
	Since $u$ is \mconc{}, we can apply the definition of \mconc{} to the bundles $Z\cup X$ and $Z\cup Y$. For every $X''\subseteq (Z\cup X)\setminus(Z\cup Y)$ with $|X''|=1$, there exists $Y''\subseteq (Z\cup Y)\setminus(Z\cup X)$ with $|Y''|\leq 1$ such that:
	\begin{align}\label{eq:u-is-mconc}
	u((Z\cup X)\setminus X''\cup Y'') + u((Z\cup Y)\setminus Y''\cup X'') \geq& u(Z\cup X)+u(Z \cup Y)
	\end{align}
	This is particularly true when $X''=X'$ from above, since $(Z\cup X)\setminus(Z\cup Y)\equiv X\setminus Y$. We can take $Y':=Y''$, since $(Z\cup Y)\setminus(Z\cup X)\equiv Y\setminus X$. It remains to prove that (\ref{eq:u-is-mconc}) implies (\ref{eq:uc-is-mconc}). 
	
	Indeed, since $X\cap Z= X'\cap Z = Y'\cap Z = \emptyset$:
	\begin{align*}
	Z\cup (X\setminus X'\cup Y') = (Z\cup X)\setminus X'\cup Y',
	\end{align*}
	since it does not matter whether we first add $X$ to $Z$ and then remove some items from the union, or first remove these items from $X$ and then add the remaining items to $Z$. Similarly, since also $Y\cap Z= \emptyset$:
	\begin{align*}
	Z\cup (Y\setminus Y'\cup X') = (Z\cup Y)\setminus Y'\cup X'
	\end{align*}
	so (\ref{eq:u-is-mconc}) and (\ref{eq:uc-is-mconc}) are equivalent.
\end{proof}
\end{fullversion}

\section{Telescopic Arrangement of Maximizing Bundles}
By definition, an agent's demanded bundles are \emph{maximizing-bundles} --- 
bundles that maximize his net-utility over all $2^m$ possible bundles. In addition to the global maximizing-bundles, we can consider the maximizing-bundles in each size-group, i.e, the maximizing-bundles among the bundles with 1 item, with 2 items, etc. In this section we prove that, when the agents' valuation is \mconc{}, the maximizing-bundles in the different size-groups have a telescopic arrangement: each maximizing-bundle contains smaller maximizing-bundles and is contained in larger maximizing-bundles.
\begin{definition}
Given valuation $u$ on $m$ items and a number $i\in\{0,\dots,m\}$, a bundle $Z_i$ is called \textbf{$i$-maximizer of $u$} if it maximizes $u$ among all bundles with $i$ items. I.e, $|Z_i|=i$ and for every other bundle $X_i$ with $i$ items, $u(Z_i)\geq u(X_i)$.
\end{definition}
\begin{lemma}\label{lem:maximizers}
For every \mconc{} valuation $u$ on $m$ items and two integers $i,j$ such that $0\leq i<j\leq m$:

(a) For every $i$-maximizer $Z_i$ there is a $j$-maximizer $Z_j'$ such that $Z_j'\supset Z_i$.

(b) For every $j$-maximizer $Z_j$ there is an $i$-maximizer $Z_i'$ such that $Z_j\supset Z_i'$.
\end{lemma}
\begin{proof}
The lemma is obviously true when $i=0$ since there is a unique 0-maximizer (the empty set). It is also true when $j=m$ since there is a unique $m$-maximizer (the set containing all items). We have to prove it for $1\leq i<j\leq m-1$, which is possible only when $m\geq 3$. The proof is by induction on $m$. 

\textbf{Base:} $m=3, j=2, i=1$. Let $Z_1$ be a 1-maximizer and $Z_2$ a 2-maximizer. If $Z_1\subseteq Z_2$ then we are done. Otherwise, $Z_1$ contains a single item, e.g. $\{x\}$, and $Z_2$ contains the other two items, $\{y,z\}$. Apply the \mconcy{} definition with $X=Z_2$ and $Y=Z_1$ and $X'=\{y\}$. Then, $Y'$ can be either $\emptyset$ or $\{x\}$:
\begin{itemize}
\item If $Y'=\emptyset$, then by the \mconcy{} condition: $u(\{z\})+u(\{x,y\})\geq u(\{x\})+u(\{y,z\})$. Then $\{z\}$ must be a 1-maximizer and $\{x,y\}$ must be a 2-maximizer; the former is contained in $Z_2$ and the latter contains $Z_1$ so we are done.
\item If $Y'=\{x\}$, then by the \mconcy{} condition: $u(\{x,z\})+u(\{y\})\geq u(\{x\})+u(\{y,z\})$. Then $\{y\}$ must be a 1-maximizer and $\{x,z\}$ must be a 2-maximizer; the former is contained in $Z_2$ and the latter contains $Z_1$ so we are done.
\end{itemize}

\textbf{Step:} we assume that the lemma is true when there are less than $m$ items and prove that it is true for $m$ items, where $m\geq 4$. Let $Z_i$ be an $i$-maximizer and $Z_j$ a $j$-maximizer. We consider several cases.

\textbf{Case 1:} There is an item which is not in $Z_i$ nor in $Z_j$. W.l.o.g. call it item 1. Let $u'$ be the restriction of $u$ to the items $\{2,\dots,m\}$. Then $Z_i$ is an $i$-maximizer of $u'$ and $Z_j$ is a $j$-maximizer of $u'$. By the induction assumption, the lemma is true for $u'$. Hence, there is a $j$-maximizer of $u'$, say $Z_j'$, which contains $Z_i$. Since both $Z_j$ and $Z_j'$ are $j$-maximizers of $u'$, $u'(Z_j)=u'(Z_j')$. Hence $u(Z_j)=u(Z_j')$. Hence, $Z_j'$ is also a $j$-maximizer of $u$, so part (a) is done. Similarly, there is an $i$-maximizer of $u'$, say $Z_i'$, which is contained in $Z_j$. Since both $Z_i$ and $Z_i'$ are $i$-maximizers of $u'$, $u'(Z_i)=u'(Z_i')$. Hence $u(Z_i)=u(Z_i')$. Hence, $Z_i'$ is also an $i$-maximizer of $u$, so part (b) is done.

\textbf{Case 2:} There is an item which is in both $Z_i$ and $Z_j$. W.l.o.g. call it item 1. Let $u'$ be the marginal valuation function $u_{\{1\}+}$. By Lemma \ref{lem:mconc-marginal}, $u'$ also is \mconc{}. It is a valuation function on $m-1$ items, $\{2,\dots,m\}$. The bundle $Z_{i-1}=Z_{i}\setminus\{1\}$ is an $(i-1)$-maximizer of $u'$ and the bundle $Z_{j-1}=Z_{j}\setminus\{1\}$ is a $(j-1)$-maximizer of $u'$. By the induction assumption the lemma is true for $u'$. Hence, there is a $j-1$-maximizer of $u'$, say $Z_{j-1}'$, which contains $Z_{i-1}$. Since both $Z_{j-1}'$ and $Z_{j-1}$ are $(j-1)$-maximizers of $u'$, $u'(Z_{j-1})=u'(Z_{j-1}')$. By definition of the marginal valuation function, this equality is equivalent to: $u(Z_j)=u(Z_{j-1}'\cup \{1\})$. Since $Z_j$ is a $j$-maximizer of $u$, $Z_{j-1}'\cup \{1\}$ is also a $j$-maximizer of $u$. It contains $Z_{i-1}\cup \{1\}=Z_i$ so part (a) is done. Similarly, there is an $i-1$-maximizer of $u'$, say $Z_{i-1}'$, which is contained in $Z_{j-1}$. Since both $Z_{i-1}$ and $Z_{i-1}'$ are $(i-1)$-maximizers of $u'$, $u'(Z_{i-1})=u'(Z_{i-1}')$. By definition of the marginal valuation function, this equality is equivalent to: $u(Z_i) =u(Z_{i-1}'\cup \{1\})$. Since $Z_i$ is an $i$-maximizer of $u$, $u(Z_{i-1}'\cup \{1\})$ is also an $i$-maximizer of $u$. It is contained in $Z_{j-1}\cup\{1\}=Z_j$ so part (b) is done.

By Case 1, the lemma is true whenever $i+j<m$.

By Case 2, the lemma is true whenever $i+j>m$. 

\textbf{Case 3:} $i+j=m$. If $i+1<j$, then $i+(i+1)<m$ and $(i+1)+j>m$. Hence, by cases 1 and 2, there is an $(i+1)$-maximizer, $Z_{i+1}'$, containing $Z_i$ and an $(i+1)$-maximizer, $Z_{i+1}''$, contained in $Z_j$. Again by cases 1 and 2, there is a $j$-maximizer $Z_j'$ containing $Z_{i+1}'$, and an $i$-maximizer $Z_i''$ contained in $Z_{i+1}''$. $Z_j'$ contains $Z_i$ and $Z_i''$ is contained in $Z_j$ so we are done.

The only case that remains is: $i+j=m$ and $i+1=j$. In that case, $m=2i+1$ (the total number of items is odd). The case $m=3, i=1, j=m-1$ was already handled in the Base, so we can assume that $m\geq 5, i\geq 2, j\leq m-2$.

Since $i+(j+1)>m$, by Case 2 part (a), there exists a $(j+1)$-maximizer, $Z_{j+1}'$, which contains $Z_i$. Also, $j+(j+1)>m$, so by Case 2 part (b), there exists a $j$-maximizer, $Z_j'$, contained in $Z_{j+1}'$. Starting at $Z_i$, we added two items to create $Z_{j+1}'$ and then removed one item to create $Z_j'$. Since $i\geq 2$, at least one item of $Z_i$ is also in $Z_j'$. Hence, $Z_i$ and $Z_j'$ are covered by Case 2. By part (a), there exists a $j$-maximizer containing $Z_i$.

Similarly, $j+(i-1)<m$. Hence, by Case 1 part (b), there exists an $(i-1)$-maximizer, $Z_{i-1}'$, contained in $Z_j$. Also, $i+(i-1)<m$, so by Case 1 part (a), there exists an $i$-maximizer, $Z_i'$, containing $Z_{i-1}'$. Starting at $Z_j$, we removed two items to create $Z_{i-1}'$ and then added one item to create $Z_i'$. Since $j\leq m-2$, at least one item not in $Z_j$ is also not in $Z_i'$. Hence, $Z_j$ and $Z_i'$ are covered by Case 1. By part (b), there exists an $i$-maximizer contained in $Z_j$.
\end{proof}

\begin{remark}
For a non-GS valuation, the ``telescopic'' property may or may not hold. For example, it holds trivially for any valuation on two item-types; it does not hold for Bob's valuation in Example \ref{exm:negative}.
\end{remark}

\section{Uniform price change}
In this section, we prove an intermediate result about the demand-flow of GS agents that may be interesting in its own right: if all items become cheaper by the same additive amount then the agent does not abandon any item, and if all items become more expensive by the same additive amount then the agent does not discover any item. 


\begin{lemma}\label{lem:mconc-d}
Let $p$ be a price-vector, $d$ a real constant, and $p'$ another price-vector such that for every item x: $p'_x = p_x + d$. If the valuation is \mconc{}, then for every bundle $P$:
	
(a) If $d\leq 0$ and $P$ is a $p$-demand, then there exists a $p'$-demand  $P'\supseteq P$.
	
(b) If $d\geq 0$ and $P$ is a $p'$-demand, then there exists a  $p$-demand $P'\supseteq P$.
\end{lemma}

\begin{proof}
We prove only part (a), since part (b) is its mirror-image. 

Let $P$ be a $p$-demand and $Q$ a $p'$-demand. We consider two cases.

\textbf{Case 1}: $|Q|\leq |P|$. Note that $u_{p'}(Q) = u_p(Q)-d\cdot|Q|$ and 
$u_{p'}(P) = u_p(P)-d\cdot|P|$, so $u_{p'}(Q) - u_p(Q) \leq u_{p'}(P) - u_p(P)$. This means that, in the move from $p$ to $p'$, $Q$ gained weakly less net-utility than $P$. Hence, if $Q$ is a $p'$-demand, $P$ is necessarily a $p'$-demand too. $P\supseteq P$ so we are done.

\textbf{Case 2}: $|Q|>|P|$. Let $i=|P|$ and $j=|Q|$. Then, $P$ is an $i$-maximizer of the net-utility function $u_p$ and $Q$ is a $j$-maximizer of the net-utility function $u_{p'}$. But, the change in price between $p$ and $p'$ does not affect the preference relation between bundles of the same size. Hence, $P$ is also an $i$-maximizer of $u_{p'}$. Since $j>i$, by Lemma \ref{lem:maximizers}/a there exists a $j$-maximizer of $u_{p'}$ that contains $P$. Call it $P'$. By definition of a $j$-maximizer, $u_{p'}(P')\geq u_{p'}(Q)$. Hence, $P'$ is also a $p'$-demand. $P'\supseteq P$ so we are done.
\end{proof}

\begin{remark}
For a non-GS valuation, the uniform-price-change property may or may not hold. For example, it holds for any valuation on two item-types, since Lemma \ref{lem:maximizers} holds in this case. It does not hold for Bob's valuation in Example \ref{exm:negative}, since when the prices change from \$10,\$10,\$10 to \$50,\$50,\$50, Bob discovers $z$.
\end{remark}

\section{Downward Demand-Flow Property (Main Result)}
\begin{theorem}\label{lem:mconc-implies-io}
	If a valuation function is \mconc{}, then it has the DDF property.
\end{theorem}
\begin{proof}
	\begin{figure}
		\hskip -1cm
		\includegraphics[scale=0.33]{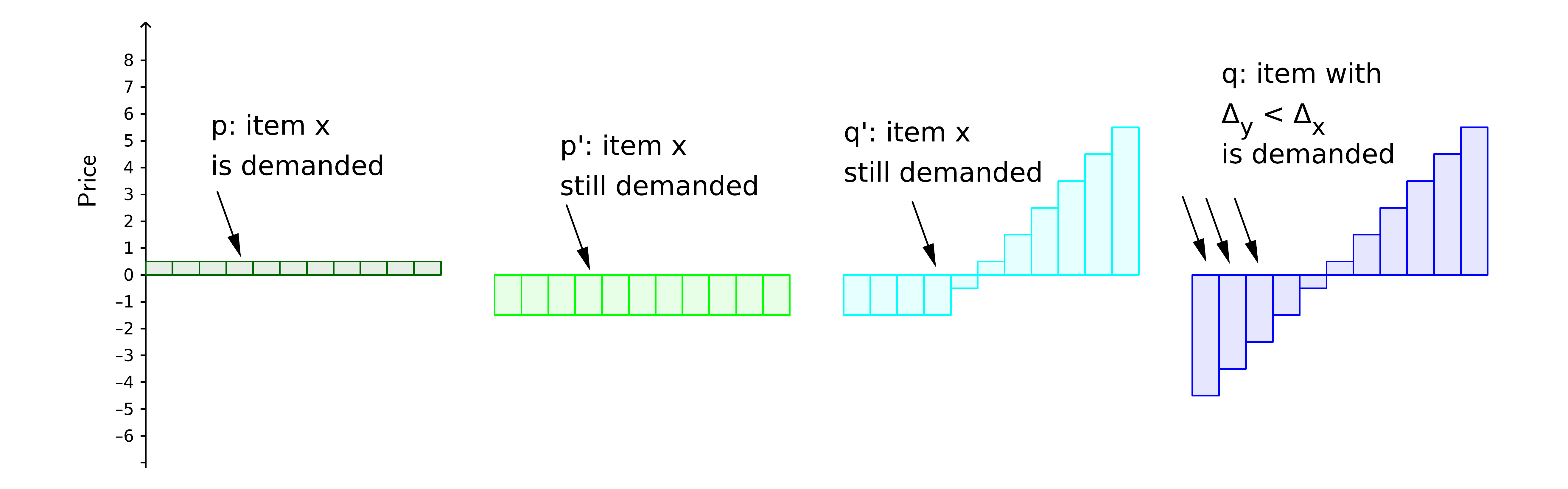}
		\caption{Illustration of prices in the proof of Lemma \ref{lem:mconc-implies-io}. Here there are 11 items and $p_x=0.5$ for all items x. The items are ordered in increasing order of $\Delta_x=q_x-p_x$.\label{fig:pricechange}}
	\end{figure}
Let $p,q$ be two price-vectors and $\Delta_x=q_x-p_x$. We now prove part (a) in the DDF definition: if the agent abandoned an item x with $\Delta_x\leq 0$, then the agent must have discovered some item y with $\Delta_y<\Delta_x$. The proof of part (b) is analogous.
	
Consider an item x with $\Delta_x\leq 0$ that is $p$-demanded but not $q$-demanded. Define a price-vector $p'$ as (see Figure \ref{fig:pricechange}):
\begin{align*}
	\forall y: p'_y = p_y + \Delta_x
\end{align*}
By Lemma \ref{lem:mconc-d}(a), all items that are $p$-demanded, including item x, are also $p'$-demanded.
	
	Define the price-vector $q'$ as (see Figure 1):
	\begin{align*}
	\Delta_y\leq \Delta_x:&& q'_y &= p'_y = p_y+\Delta_x \\
	\Delta_y\geq \Delta_x:&& q'_y &= p'_y + (\Delta_y-\Delta_x) = p_y+\Delta_y = q_y
	\end{align*}
	Between $p'$ and $q'$, the prices of items above x weakly increased while the prices of item x and the items below x remained the same. By the GS property, item x is $q'$-demanded, and all items below x that were $p$-demanded are $q'$-demanded.
	
	The last step of the proof --- the move from $q'$ to $q$ --- is true for arbitrary valuations (not only GS). Since x was $q'$-demanded, there was a $q'$-demand $Q'$ that contained x. Since x is not $q$-demanded, $Q'$ is not a $q$-demand. This means that there must be a different $q$-demand, say $Q$, that became more attractive than $Q'$, i.e, $u_q(Q)>u_q(Q')$. But $u_{q'}(Q)\leq u_{q'}(Q')$, so necessarily, in the move from $q'$ to $q$, the bundle $Q$ became cheaper more than $Q'$. Since the only items that became cheaper from $q'$ to $q$ are items with $\Delta_y<\Delta_x$, the bundle $Q$ must contain at least one of these items y which was not previously demanded. This implies that our agent, who abandoned x, has discovered y.
\end{proof}

\begin{fullversion}
\section{Remark}
\citet{Gul1999Walrasian} prove that GS is equivalent to two other properties: Single Improvement (SI) and No Complementaries (NC). They also present a property which they call Strong No Complementaries (SNC) and prove that it implies NC. They do not prove that NC implies SNC, but they also do not prove otherwise, i.e, they do not give an example of a valuation that is NC and not SNC. Based on our failure to find such an example ourselves, and based on the similarity between the SNC condition and the MC condition, we conjecture that SNC is actually equivalent to MC (and hence, to NC and SI and GS).
\end{fullversion}

\paragraph{\bf Acknowledgements} This research was funded in part by the ISF grants 1083/13 and 1224/12, the Doctoral Fellowships of Excellence Program and the Mordecai and Monique Katz Graduate Fellowship Program at Bar-Ilan University. 
We are grateful to Assaf Romm and Elizabeth Baldwin and an anonymous referee for their helpful comments.

\bibliographystyle{elsarticle-num-names}

\end{document}